\documentclass{article}

\usepackage{microtype}

\usepackage{amsthm}
\usepackage{euscript}
\usepackage{amsopn,amsfonts}
\usepackage{float,xspace}
\usepackage{amsmath}
\usepackage{amssymb}
\usepackage{algorithm2e}
\usepackage{paralist}
\usepackage{hyperref}
\usepackage{authblk}
\usepackage{fullpage}
\usepackage{graphicx}
\usepackage{url}
\def\triangleq{\overset{\mathrm{def}}{=}}

\let\eps=\varepsilon

\theoremstyle{plain}
\newtheorem{theorem}{Theorem}
\newtheorem{claim}[theorem]{Claim}
\newtheorem{notation}[theorem]{Notation}
\newtheorem{note}[theorem]{Note}
\newtheorem{proposition}[theorem]{Proposition}
\newtheorem{mycorollary}[theorem]{Corollary}
\newtheorem{mylemma}[theorem]{Lemma}
\newtheorem{remark}[theorem]{Remark}

\newcommand{\R}{\mathbb{R}}

\newcommand{\mkvc}{\textsf{Vertex Cover}\xspace}
\newcommand{\I}{\ensuremath{{\cal I}}}
\newcommand{\ind}[1]{\mathbf{e}_{#1}}

\def\ff{\setfam{F}}

\newcommand{\setfam}[1]{\EuScript{#1}}
\newcommand{\probs}[2]{\mathrm{Prob}_{#1}\big[ #2 \big]}
\newcommand{\expcts}[2]{\mathbb{E}_{#1}\big[ #2 \big]}

\newcommand{\xvec}{\mathrm{x}}

\def\sm{\setminus}
\def\indset{\operatorname{IS}}
\def\rindset{\operatorname{rel-IS}}
\def\vercov{\operatorname{VC}}
\def\rvercov{\operatorname{rel-VC}}
\def\minvc{\mkvc}

\begin{document}

\title{The Hardness of Approximation of Euclidean k-means}

\author[1]{Pranjal Awasthi}
\author[2]{Moses Charikar}
\author[3]{Ravishankar Krishnaswamy}
\author[4]{Ali Kemal Sinop}
\affil[1]{Computer Science Department,\\
  Princeton University,\\
  \texttt{pawasthi@cs.cmu.edu}}
\affil[2]{Computer Science Department\\
  Princeton University\\
  \texttt{moses@cs.princeton.edu}}
\affil[3]{Microsoft Research\\
  \texttt{ravishan@cs.cmu.edu}}
\affil[4]{Simons Institute for Theory of Computing \\
	University of California, Berkeley\\
	\texttt{asinop@cs.cmu.edu}}

\maketitle

\begin{abstract}
  The Euclidean $k$-means problem is a classical problem that has been
  extensively studied in the theoretical computer science, machine
  learning and the computational geometry communities. In this
  problem, we are given a set of $n$ points in Euclidean space $\R^d$,
  and the goal is to choose $k$ center points in $\R^d$ so that the
  sum of squared distances of each point to its nearest center is
  minimized.
  The best approximation algorithms for this problem include a
  polynomial time constant factor approximation for general $k$ and a
  $(1+\epsilon)$-approximation which runs in time poly(n)
  exp($k/\epsilon$). %
  At the other extreme, the only known computational complexity result
  for this problem is NP-hardness~\cite{Aloise09}.  The main
  difficulty in obtaining hardness results stems from the Euclidean
  nature of the problem, and the fact that any point in $\R^d$ can be
  a potential center. This gap in understanding left open the
  intriguing possibility that the problem might admit a PTAS for all
  $k,d$.

  In this paper we provide the first hardness of approximation for the
  Euclidean $k$-means problem. Concretely, we show that there exists a
  constant $\epsilon > 0$ such that it is NP-hard to approximate the
  $k$-means objective to within a factor of $(1+\epsilon)$.  We show
  this via an efficient reduction from the vertex cover problem on
  \emph{triangle-free graphs}: given a triangle-free graph, the goal
  is to choose the fewest number of vertices which are incident on all
  the edges.
  Additionally, we give a proof that the current best hardness results
  for vertex cover can be carried over to triangle-free graphs. To
  show this we transform $G$, a known hard vertex cover instance, by
  taking a graph product with a suitably chosen graph $H$, and showing
  that the size of the (normalized) maximum independent set is almost
  exactly preserved in the product graph using a spectral analysis,
  which might be of independent interest.
\end{abstract}

\section{Introduction}
Clustering is the task of partitioning a set of items such as web
pages, protein sequences etc. into groups of related items. This is a
fundamental task in machine learning, information retrieval,
computational geometry, computer vision, data visualization and many
other domains. In many applications, clustering is often used as a
first step toward other fine grained tasks such as
classification. Needless to say, the problem of clustering has
received significant attention over the years and there is a large
body of work on both the applied and the theoretical aspects of the
problem~\cite{AGK+04,ARR98,CGTS99,FKKR03,jms06,Kanungo02,KSS04,Lloyd06,BBG09,
  Lloyd, top10}.
A common way to approach the task of clustering is to map the set of
items into a metric space where distances correspond to how different
two items are from each other. Using this distance information, one
then tries to optimize an objective function to get the desired
clustering. Among the most commonly used objective function used in
the clustering literature is the $k$-means objective function. In the
$k$-means problem, the input is a set $S$ of $n$ data points in
Euclidean space $\R^d$, and the goal is to choose $k$ center points
$C^* = \{c_1, c_2, \ldots, c_k\}$ from $\R^d$ so as to minimize
$\Phi = \sum_{x \in S} \min_i \| x - c(x) \|^2$, where $c(x) \in C^*$
is the center closest to $x$.
Aside from being a natural clustering objective, an important
motivation for studying this objective function stems from the fact
that a very popular and widely used heuristic %
(appropriately called the \emph{$k$-means heuristic}~\cite{Lloyd}) %
attempts to minimize this $k$-means objective function.

While the $k$-means heuristic is very much tied to the $k$-means
objective function, there are many examples where it converges to a
solution which is far away from the optimal $k$-means solution. This
raises the important question of whether there exist provable
algorithms for the $k$-means problem in general Euclidean space, which
is the focus problem of our paper. Unfortunately though, the
approximability of the problem is not very well understood. From the
algorithmic side, there has been much focus on getting
$(1+\epsilon)$-approximations that run as efficiently as
possible. Indeed, for fixed $k$, Euclidean $k$-means admits a
PTAS~\cite{KSS04, FeldmanMS07}. These algorithms have exponential
dependence in $k$, but only linear dependence in the number of points
and the dimensionality of the space. As mentioned above, there is also
empirical and theoretical evidence for the effectiveness of very
simple heuristics for this problem~\cite{Lloyd06, Lloyd, KK10}. For
arbitrary $k$ and $d$, the best known approximation algorithm for
$k$-means achieves a factor of $9+\epsilon$~\cite{Kanungo02}. In
contrast to the above body of work on getting algorithms for
$k$-means, lower bounds for $k$-means have remained elusive. In fact,
until recently, even NP-hardness was not known for the $k$-means
objective~\cite{Dasgupta08, Aloise09}. This is perhaps due to the fact
that as opposed to many discrete optimization problems, the $k$-means
problem allows one to choose any point in the Euclidean space as a
center. The above observations lead to the following intriguing
possibility

\textit{Is there a PTAS for Euclidean $k$-means for arbitrary $k$ and dimension $d$?}

In this paper we answer this question in the negative and provide the
first hardness of approximation for the Euclidean $k$-means problem.

\begin{theorem} \label{thm:main} There exists a constant $\epsilon>0$
  such that it is NP-hard to approximate the Euclidean $k$-means to a
  factor better than $(1+\epsilon)$.
\end{theorem}

The starting point for our reduction is the \textsf{Vertex-Cover}
problem on triangle-free graphs: here, given a triangle-free graph,
the goal is to choose the fewest number of vertices which are incident
on all the edges in the graph. This naturally leads us to our other
main result in this paper, that of showing hardness of approximation
of vertex cover on triangle-free graphs. Kortsarz et al~\cite{K08}
show that if the vertex cover problem is hard to approximate to a
factor of $\alpha \geq 3/2$, then it is hard to approximate vertex
cover on triangle-free graphs to the same factor of $\alpha$. While
such a hardness (in fact, a factor of $2- \epsilon$~\cite{Khot08}) is
known assuming the stronger unique games conjecture, the best known
NP-hardness results do not satisfy $\alpha \geq 3/2$. We settle this
question by showing NP-hardness results for approximating vertex cover
on triangle-free graphs, which match the best known hardness on
general graphs.

\begin{theorem}
  It is NP-hard to approximate Vertex Cover on triangle-free graphs to
  within any factor smaller than $1.36$.
\end{theorem}

\section{Main Technical Contribution}
In Section~\ref{sec:vc-kmeans}, we show a reduction from
\textsf{Vertex-Cover} on triangle-free graphs to Euclidean $k$-means
where the vertex cover instances have small cover size if and only if
the corresponding $k$-means instances have a low cost. A crucial
ingredient is to relate the cost of the clusters to the structural
properties of the original graph, which lets us transition from the
Euclidean problem to a completely combinatorial problem. Then in
Section~\ref{sec:vc-girth}, we prove that the known hardness of
approximation results for \textsf{Vertex-Cover} carry over to
triangle-free graphs. This improves over existing hardness results for
vertex cover on triangle-free graphs~\cite{K08}. Furthermore, we
believe that our proof techniques are of independent
interest. Specifically, our reduction transforms known hard instances
$G$ of vertex cover, by taking a graph product with an appropriately
chosen graph $H$. We then show that the size of the vertex cover in
the new graph (in proportion to the size of the graph) can be related
to spectral properties of $H$. In fact, by choosing $H$ to have a
bounded spectral radius, we show that the vertex covers in $G$ and the
product graph are roughly preserved, while also ensuring that the
product graph is triangle-free.  Combining this with our reduction to
$k$-means completes the proof.

\section{Related Work}
Arthur and Vassilvitskii~\cite{Arthur07} proposed $k$-means++, a
random sampling based approximation algorithm for Euclidean $k$-means
which achieves a factor of $O(\log k)$. This was improved by Kanungo
et al.~\cite{Kanungo02} who proposed a local search based algorithm
which achieves a factor of $(9+\epsilon)$. This is currently the best
known approximation algorithm for $k$-means. For fixed $k$ and $d$,
Matousek~\cite{Matousek00} gave a PTAS for $k$-means which runs in
time $O(n\epsilon^{-2k^2 d} \log^k n))$. Here $n$ is the number of
points and $m$ is the dimensionality of the space. This was improved
by Badoiu et al.~\cite{Badoiu02} who gave a PTAS for fixed $k$ and any
$d$ with run time $O(2^{({k/\epsilon})^{O(1)}}poly(d) n \log^k n)$.
Kumar et al.~\cite{KSS04} gave an improved PTAS with exponential
dependence in $k$ and only linear dependence in $n$ and $d$. Feldman
et al.~\cite{FeldmanMS07} combined this with efficient coreset
constructions to give a PTAS for fixed $k$ with improved dependence on
$k$. The work of Dasgupta~\cite{Dasgupta08} and Aloise et
al.~\cite{Aloise09} showed that Euclidean $k$-means is NP-hard even
for $k=2$. Mahajan et al.~\cite{Mahajan09} also show that the
$k$-means problem is NP-hard for points in the plane.

There are also many other clustering objectives related to $k$-means
which are commonly studied. The most relevant to our discussion are
the $k$-median and the $k$-center objectives. In the first problem,
the objective is to pick $k$ centers to minimize the sum of distances
of each point to the nearest center (note that the distances are not
squared). The problem deviates from $k$-means in two crucial aspects,
both owing to the different contexts in which the two problems are
studied: (i) the $k$-median problem is typically studied in the
setting where the centers are one of the data points (or come from a
set of possible centers specified in the input), and (ii) the problem
is also very widely studied on general metrics, without the Euclidean
restriction. The $k$-median problem has been a testbed of developing
new techniques in approximation algorithms, and has constantly seen
improvements even until very
recently~\cite{jv01,jms06,Li13}. Currently, the best known
approximation for $k$-median is a factor of $2.611 + \epsilon$ due to
Bykra et al.~\cite{Bykra14}. On the other hand, it is also known that
the $k$-median objective (on general metrics) is NP-hard to
approximate to a factor better than $(1+1/e)$~\cite{jms06}. When
restricted to Euclidean metrics, Kolliopoulos et al.~\cite{KR99} show
a PTAS for $k$-median on constant dimensional spaces. On the negative
side for $k$-median on Euclidean metrics, it is known that the
discrete problem (where centers come from a specified input) cannot
have a PTAS under standard complexity assumptions~\cite{Indyk03}. As
mentioned earlier, all these results are for the version when the
possible candidate centers is specified in the input. For the problem
where any point can be a center, Arora et al.~\cite{ARR98} show a PTAS
when the points are on a 2-dimensional plane.

In the $k$-center problem the objective is to pick $k$ center points
such that the maximum distance of any data point to the closest center
point is minimized. In general metrics, this problem admits a
$2$-factor approximation which is also optimal assuming
$P \neq NP$~\cite{hochbaum1986unified}. For Euclidean metric when the
center could be any point in the space, the upper bound is still $2$
and the best hardness of approximation is a factor
$1.82$~\cite{feder1988optimal}.

\section{Our Hardness Reduction: From Vertex Cover to Euclidean
  $k$-means} \label{sec:vc-kmeans}
In this section, we show a reduction from the Vertex-Cover problem (on
triangle-free graphs) to the $k$-means problem. Formally, the vertex
cover problem can be stated as follows: Given an undirected graph
$G = (V,E)$, choose a subset $S$ of vertices (with minimum $|S|$) such
that $S$ is incident on every edge of the graph. More specifically,
our reduction establishes the following theorem.

\begin{theorem} \label{thm:redux} There is an efficient reduction from
  instances of Vertex Cover (on triangle-free graphs) to those of
  Euclidean $k$-means that satisfies the following properties:
  \begin{enumerate}
  \item [(i)] if the \mkvc instance has value $k$, then the $k$-means
    instance has cost at most $m-k$.
  \item [(ii)] if the \mkvc instance has value at least
    $k (1 + \epsilon)$, then the optimal $k$-means cost is at least
    $m - (1 - \Omega(\epsilon))k$. Here, $\epsilon$ is some fixed
    constant $ > 0$.
  \end{enumerate}
\end{theorem}

In \ref{sec:girth}, we show that there exist triangle-free graph
instances of vertex cover on $m = \Theta(n)$ edges, and
$k = \Omega(n)$ such that it is NP-hard to distinguish if the instance
has a vertex cover of size at most $k$, or all vertex covers have size
at least $(1+\epsilon)k$, for some constant $\epsilon > 0$.

Now, let $k = m/\Delta$ where $\Delta = \Omega(1)$ from the hard
vertex cover instances. Then, from~\ref{thm:redux}, we get that if the
vertex cover has value $k$, then the $k$-means cost is at most
$m (1 - \frac{1}{\Delta})$, and if the vertex cover is at least
$k(1+\epsilon)$, then the optimal $k$-means cost is at least
$m (1 - \frac{1 - \Omega(\epsilon)}{\Delta})$. Therefore, the vertex
cover hardness says that it is also NP-hard to distinguish if the
resulting $k$-means instance has cost at most
$m ( 1 - \frac{1}{\Delta})$ or cost more than
$m (1 - \frac{1 - \Omega(\epsilon)}{\Delta})$. Since $\Delta$ is a
constant, this implies that it is NP-hard to approximate the $k$-means
problem within some factor $(1 +\Omega(\epsilon))$, thereby
establishing our main result~\ref{thm:main}. In what follows, we
prove~\ref{thm:redux}.

\subsection{Proof of \texorpdfstring{\ref{thm:redux}}
  {Theorem~\ref{thm:redux}} }
\label{pf:redux} Let $G = (V,E)$ denote the graph in the \mkvc
instance \I, with parameter $k$ denoting the number of vertices we can
select.
We associate the vertices with natural numbers $[n]$. Therefore, we
refer to vertices by natural numbers $i$, and edges by pairs of
natural numbers $(i,j)$.

\medskip \noindent {\bf Construction of k-means Instance $\I_{km}$.}
For each vertex $i \in [n]$, we have a unit vector
$\xvec_i = (0, 0, \ldots, 1, \ldots, 0)$ which has a $1$ in the
$i^{th}$ coordinate and $0$ elsewhere.  Now, for each edge
$e \equiv (i,j)$, we have a vector $\xvec_e \triangleq e_i + e_j$.
Our data points on which we solve the $k$-means problem is precisely
$\{ \xvec_e \, : \, e \in E\}$.  This completes the definition of
$\I_{km}$.

\begin{note}
  As stated, the dimensionality of the points we have constructed is
  $n$, and we get a hardness factor of $(1+\epsilon)$. However, by
  using the dimensionality reduction ideas of Johnson and
  Lindenstrauss (see, e.g. ~\cite{dasgupta2003elementary}), without
  loss of generality, we can assume that the points lie in
  $O(\log n/\epsilon^2)$ dimensions and our hardness results still
  hold true. This is because, after the transformation, all pairwise
  distances (and in particular, the $k$-means objective function) are
  preserved upto a factor of $(1+\epsilon/10)$ of the original values,
  and so our hardness factor is also (almost) preserved, i.e., we
  would get hardness of approximation of $(1+\Omega(\epsilon))$.

  However, for simplicity, we stick with the $n$ dimensional vectors
  as it makes the presentation much cleaner.
\end{note}

\subsection{Completeness}
Suppose $\I$ is such that there exists a vertex cover
$S^* = \{v_1, v_2, \ldots, v_k\}$ of $k$ vertices which can cover all
the edges.  We will now show that we can recover a good clustering of
low $k$-means cost.  To this end, let $E_{v_\ell}$ denote the set of
edges which are covered by $v_{\ell}$ for $1 \leq \ell \leq k$. If an
edge is covered by two vertices, we assume that only one of them
covers it. As a result, note that the $E_{v_\ell}$'s are pairwise
disjoint (and their union is $E$), and each $E_{v_{\ell}}$ is of the
form
$\{(v_{\ell}, w_{\ell, 1}), (v_{\ell}, w_{\ell, 2}), \ldots,
(v_{\ell}, w_{\ell, p_\ell})\}$.

Now, to get our clustering, we do the following: for each $v \in S^*$,
form a cluster out of the data points
$\ff_{v} := \{\xvec_{e} \, : \ e \in E_{v}\}$. We now analyze the
average connection cost of this solution.  To this end, we begin with
some easy observations about the k-means clustering. Indeed, since any
cluster is of a set of data points (corresponding to a subset of edges
in the graph $G$), we shall abuse notation and associate any cluster
$\ff$ also with the corresponding subgraph on $V$, i.e.,
$\ff \subseteq E$. Moreover, we use $d_\ff(i)$ to denote the degree of
node $i$ in $\ff$ and $m_{\ff}$ to denote the number of edges in
$\ff$, $m_{\ff} = |\ff|$. Finally, we refer by $d_G(i)$ the degree of
vertex $i$ in $G$.

\begin{claim}
  Given any clustering $\{ \ff \}$, the following hold.
  \begin{enumerate}
  \item [(i)] $\sum_{\ff} d_\ff(i) = d_G(i)$.
  \item [(ii)] $\sum_i \sum_{\ff} d_\ff(i) = 2m = 2|E|$.
  \end{enumerate}
\end{claim}
\begin{proof}
  The proof is immediate, because every edge $e \in E$ belongs to
  exactly one cluster in $\{\ff\}$.
\end{proof}

Our next claim relates the connection cost of any cluster $\ff$ to the
structure of the associated subgraph, which forms the crucial part of
the analysis.
\begin{claim}
  \label{cl:cost}
  The total connection cost of any cluster $\ff$ is:
  \[\sum_i d_\ff(i) (1 - \frac{1}{m_{\ff}} d_\ff(i)).\]
\end{claim}
\begin{proof}
  Firstly, note that $\sum_i d_\ff(i) = 2 m_{\ff}$. Now consider the
  center $\mu_{\ff}$ of cluster $\ff$.  By definition, we have that at
  coordinate $i \in V$:
  \[
  \mu_{\ff}(i) = \frac{1}{m} \sum_{S \in \ff: i\in S} 1 =
  \frac{d_\ff(i)}{m}.
  \]
  So $\|\mu_{\ff}\|^2 = \frac{1}{m^2} \sum_i d_\ff(i)^2$.  Hence the
  total cost of this clustering is:
  \begin{align*}
    c_{\ff} =&
               \sum_{e \in \ff} (\|\xvec_e - \mu_{\ff}\|^2)
               = \sum_{e \in \ff} (\|\xvec_e\|^2 - \|\mu_{\ff}\|^2) \\
    = & 2 m_\ff  - \frac1m \sum_{i \in V} d_\ff(i)^2
        = \sum_{i} d_\ff(i) - \frac1m d_\ff(i)^2.
  \end{align*}
  The first equality here uses the fact that
  $m_\ff \mu_{\ff} = \sum_{e \in \ff} \xvec_e$; and the second
  equality uses the fact that $\| \xvec_e \|^2 = 2$ for each data
  point.
\end{proof}

\begin{claim} \label{cl:yescost} There exists a clustering of our
  $k$-means instance $\I_{km}$ with cost at most $m - k$, where $m$ is
  the number of edges in the graph $G = (V,E)$ associated with the
  vertex cover instance $\I$, and $k$ is the size of the optimal
  vertex cover.
\end{claim}

\begin{proof}
  Consider a cluster $\ff_v$, which consists of data points associated
  with edges covered by a single vertex $v$. Then, by \ref{cl:cost},
  the connection cost of this cluster is precisely $m_{\ff_v} - 1$,
  since the sub-graph associated with a cluster is simply a star
  rooted at $v$. Here, $m_{\ff_v}$ is the number of edges which $v$
  covers in the vertex cover (if an edge is covered by different
  vertices in the cover, it is included in only one vertex). Then,
  summing over all clusters, we get the claim.
\end{proof}

\subsection{Soundness}
In this section, we show that if there is a clustering of low
$k$-means cost, then there is a very good vertex cover for the
corresponding graph. We begin with some useful notation.

\def\gcost{\mathsf{Cost}}
\begin{notation}\label{def:graph-var}
  Given a set $E' \subseteq \binom{V}{2}$ of $m_{E'} = | E' |$ edges
  with corresponding node degrees $(d_1,\ldots, d_n)$, we define
  $\gcost(E')$ as the following:
  \[
  \gcost(E') \triangleq \sum_{u\in V} d_u \Big(1 -
  \frac{d_u}{m_{E'}} \Big).
  \]
\end{notation}
Note that, by~\ref{cl:cost}, the connection cost of a clustering
$\Gamma = \{ \ff_1, \ff_2, \ldots, \ff_k \}$ of the $n$ points is
equal to $\sum_{i} \gcost(\ff_i)$. Recall that we abuse notation
slightly and view each cluster $\ff_i$ of the data points also as a
subset of $E$. Moreover, because $\Gamma$ clusters all points, the
subgraphs $\ff_1, \ff_2, \ldots, \ff_k$ form a partition of $E$. Using
this analogy, we study the properties of each subgraph and show that
if the $k$-means cost of $\Gamma$ is small, then most of these
subgraphs in fact are stars. This will in turn help us recover a small
vertex cover for $G$. We begin with a simple property of $\gcost(E')$.

\begin{proposition}
  \label{thm:graph-var-props}
  For any set of $m_{E'}$ edges $E'$,
  $
  m_{E'}-1\le \gcost(E') \le 2 m_{E'}-1.
  $
\end{proposition}
\begin{proof}
  We have
  $\gcost(E') = \sum_{u\in V} d_u \Big(1 - \frac{d_u}{m_{E'}} \Big) =
  2m_{E'} - \frac{\sum_{u \in V} d^2_u}{m_{E'}}$.
  The proof follows from noting that
  $\frac{\sum_{u \in V} d^2_u}{m_{E'}} \geq
  \frac{\sum_{u \in V} d_u}{m_{E'}} = 2$
  and $\frac{\sum_{u \in V} d^2_u}{m_{E'}} \leq m_{E'}+1$. The last
  inequality is due to the fact that $\sum_{u \in V} d^2_u$ is
  maximized by the degree sequence $(m_{E'},1,1,\ldots ,1)$.
\end{proof}
\begin{theorem}
  \label{thm:kmeans-to-vc}
  If the $k$-means instance $\I_{km}$ has a clustering
  $\Gamma = \{\ff_1,\ldots, \ff_k\}$
  with $\sum_{\ff\in \Gamma} \gcost(\ff) \le m - (1-\delta) k$, then
  there exists a $(1+O(\delta)) k$-vertex cover of $G$ in the instance
  $\I$.
\end{theorem}
\def\goodc{\mathrm{Good}}
\def\badc{\mathrm{Bad}}
\def\uglyc{\mathrm{Ugly}}
\def\cover{\mathrm{Cover}_0}
Note that this, along with~\ref{cl:yescost} would complete the proof
of~\ref{thm:redux}.
\begin{proof}
  For each $i \in [k]$, let $m_i \triangleq |\ff_i|$ and
  $\nu_i \triangleq \sum_u d_u(\ff_i)^2$.  Note that
  $\gcost(\ff_i) = 2 m_i - \frac{\nu_i}{m_i}$.  By
  \ref{thm:graph-var-props}, each $i \in [k]$ satisfies
  $m_i - 1\le \gcost(\ff_i) \le 2 m_i - 1$.  Hence if we define
  $\delta_i$ as $\delta_i \triangleq \gcost(\ff_i) - (m_i - 1)$, then
  $0\le \delta_i \le m_i$. Moreover
  $ \frac{\nu_i}{m_i} = m_i + 1 - \delta_i.  $
  Thus:
  \[
  m - (1-\delta) k \ge \sum_i \gcost(\ff_i) = \sum_i (\delta_i + m_i -
  1) = \sum_i \delta_i + m - k \implies \delta k \ge \sum_i \delta_i.
  \]
  This means, except $\le 2 \delta k$ clusters, the remaining clusters
  all have $\delta_i \le \frac12$.  Moreover,
  \ref{thm:small-cover-all} implies all these $(1-2\delta) k$ clusters
  are either stars or triangles and have $\delta_i = 0$. Since the
  graph is triangle free, they are all stars, and hence the
  corresponding center vertices cover all the edges in the respective
  clusters. It now remains to cover the edges in the remaining
  $2 \delta k$ clusters which have larger $\delta_i$ values. Indeed,
  even for these clusters, we can appeal to \ref{thm:small-cover-all},
  and choose \emph{two} vertices per cluster to cover all but
  $\delta_i$ edges in each cluster. So the size of our candidate
  vertex cover is at most $k(1 +2\delta)$, and we have covered all but
  $\sum_{i} \delta_i$ edges. But now, we notice that
  $\sum_{i} \delta_i \leq \delta k$, and so we can simply include one
  vertex per uncovered edge and would obtain a vertex cover of size at
  most $k (1 + 3 \delta)$, thus completing the proof.
\end{proof}

\begin{mylemma}\label{thm:small-cover-all}
  Given a graph $G_\ff=(V,\ff)$ with $m = |\ff|$ edges and degrees
  $(d_1,\ldots, d_n)$; let $\delta$ be such that
  \[
  \frac{1}{m} \sum_u {d_u}^2 = m + 1 - \delta.
  \]
  There always exists an edge $\{u,v\}\in \ff$ with
  $d_u + d_v - 1 \ge m+1+\delta$.  Furthermore, if $\delta < \frac12$,
  then $\delta=0$ and $G_\ff$ is either a star graph or a triangle.
\end{mylemma}
\begin{proof}
  Since $\sum_u {d_u}^2 = \sum_{u\sim v} (d_u + d_v)$, we can think of
  $\frac{1}{m} \sum_u {d_u}^2$ as the the expectation of $d_u + d_v$
  over a random edge chosen uniformly, $\{u,v\} \in E$:
  \[
  \frac{1}{m} \sum_u {d_u}^2 = \expcts{u\sim v}{d_u + d_v}.
  \]
  From this, we can immediately conclude the existence
  of an edge
  $\{u,v\}$ with $d_u + d_v \ge m + 1 - \delta$.
  Now to complete the second part of the Lemma statement, suppose $d_u \ge d_v$.
  The number of edges incident to $\{u,v\}$ is:
  \[
  d_u + d_v - 1 \ge m - \delta
  \overset{\delta<1}{\implies}
  d_u + d_v - 1 = m.
  \]
  So all edges are incident to $u$ or $v$, and
  $d_w \le 2$ if $w \notin \{u,v\}$.
  If $d_v \le 1$, then we are done.
  In the other case, we have $d_v \ge 2 \ge d_w$
  for all $w \notin \{u,v\}$.
  Let $\alpha \triangleq d_u$ and $\beta \triangleq d_v$.
  The degree sequence $(d_1,\ldots,d_n)$ is strongly
  majorized by the following sequence, $d'$:
  \[
  d' \triangleq
  \Big(\alpha, \beta,
  \underbrace{2, \ldots, 2,}_{\mbox{$\beta-1$ many}}
  \overbrace{1, \ldots, 1}^{\mbox{$\alpha-\beta$ many}}
  \Big).
  \]
  Since $\sum_u d_u^2$ is Schur-convex, its value
  increases under majorization:
  \begin{align*}
    (\alpha + \beta-1) (\alpha+\beta-\delta)= &
                                                m ( m + 1 - \delta) \le
                                                \sum_u d_u^2
                                                \le  \sum_u {d_u'}^2 \\
    = & \alpha^2 + \beta^2
        + 4 (\beta - 1) + (\alpha - \beta). \\
    \implies 0 \le & (\alpha+\beta-1) \delta+2 \alpha+4
                     \beta-4-2
                     \alpha \beta \\
    = &  (\alpha+\beta-1) \delta+2 \alpha ( 1 - \beta)
        +4 (\beta-1).
  \end{align*}
  So we obtain
  $
  2 \alpha ( \beta - 1 ) \le
  (\alpha+\beta-1) \delta+4 (\beta - 1)
  $. Since $\beta \ge 2$, 
  we divide both sides by $\beta-1$:
  \[
  2 \alpha \le \frac{\alpha}{\beta-1} \delta
  + 4 + \delta
  \le \delta \alpha  + 4 + \delta.
  \]
  In particular,
  $(2-\delta) \alpha \le 4 + \delta
  \implies
  \alpha \le \frac{4 + \delta}{2 - \delta} < 3$
  as $\delta < 1/2$. Hence $\alpha \le 2$.	
  Consequently, $d_u = d_v = 2$ and
  $m = d_u + d_v - 1 = 3$.
  There are two possible cases:
  The graph is either a $3$-cycle or $4$-path.
  In the latter case, the corresponding $\delta$ 
  is:
  \[
  \delta = m + 1 - 
  \frac1m \sum_u {d_u}^2 
  = 4 - 
  \frac{1}{3} (2^2 + 2^2 + 1 + 1 )
  = 4 - \frac{10}{3} = \frac{2}{3} > \frac12;
  \]  which is a contradiction and the graph is a triangle.
\end{proof}

Putting the pieces together, we get the proof of \ref{thm:redux}.

\remark{\textbf{Unique Games Hardness:}} Khot and Regev~\cite{Khot08}
show that approximating Vertex-Cover to factor $(2-\epsilon)$ is hard
assuming the Unique Games conjecture. Furthermore, Kortsarz et
al.~\cite{K08} show that any approximation algorithm with ratio
$\alpha \geq 1.5$ for Vertex-Cover on $3$-cycle-free graphs implies an
$\alpha$ approximation algorithm for Vertex-Cover (on general graphs).
This result combined with the reduction in this section immediately
implies APX hardness for $k$-means under the unique games
conjecture. In the next section we generalize the result of Kortsarz
et al.~\cite{K08} by giving an approximation preserving reduction from
Vertex-Cover on general graphs to Vertex-Cover on triangle-free
graphs. This would enable us to get APX hardness for the $k$-means
problem.

\section{Hardness of Vertex Cover on Triangle-Free Graphs} \label{sec:vc-girth}
\label{sec:girth}
In this section, we show that the \mkvc problem is as hard on
triangle-free graphs %
as it is on general graphs. To this end, for any graph $G=(V,E)$, we
define $\indset(G)$ as the size of maximum independent set in $G$.
For convenience, we define $\rindset(G)$
as the ratio of $\indset(G)$ to the number of nodes in $G$:
\[
\rindset(G) \triangleq \frac{\indset(G)}{|V|}.
\]
Similarly, let $\vercov(G)$ be the
size of minimum vertex cover in $G$ and
$\rvercov(G)$ be the ratio $\frac{\vercov(G)}{|V|}$.
The following is well known, which says independent sets
and vertex covers are duals of each other.
\begin{proposition}
  \label{thm:vc-is-dual}
  Given $G=(V,E)$, $I \subseteq V$ is an independent set
  if and only if $C = V\sm I$ is a vertex cover. In particular,
  $\indset(G) + \vercov(G) = |V|$.
\end{proposition}
We will prove the following theorem.
\begin{theorem}\label{thm:is-tri-free-apx-preserve}
  For any constant $\eps>0$, 
  there is a $(1+\eps)$-approximation-preserving 
  reduction for independent set from any graph $G=(V,E)$ 
  with maximum degree $\Delta$ 
  to triangle-free graphs with 
  $\mathrm{poly}(\Delta,\eps^{-1}) |V|$ nodes and degree 
  $\mathrm{poly}(\Delta,\eps^{-1})$ in deterministic
  polynomial time.	
\end{theorem}
Combining~\ref{thm:is-tri-free-apx-preserve}
with the best known unconditional hardness result 
for \mkvc, due to Dinur~and~Safra~\cite{ds05},
we obtain the following corollary.
\begin{mycorollary}
  \label{thm:vc-tri-free-hardness}
  Given any unweighted 
  triangle-free 
  graph $G$ with bounded degrees,
  it is NP-hard to approximate \minvc\ 
  within any factor smaller than $1.36$.
\end{mycorollary}
Given two simple
graphs $G=(V_1, E_1)$ and $H=(V_2, H_2)$,
we define the Kronecker product of $G$ and $H$,
$G\otimes H$, as the
graph with nodes $V(G\otimes H) = V_1 \times V_2$
and edges:
\[
E(G\otimes H) = \Big\{ \{(u,i),(v,j)\} \big| \{u,v\}\in E(G),\
\{i,j\}\in E(H) \Big\}.
\]
Observe that, if $A_G$ and $A_H$ denote the adjacency
matrix of $G$ and $H$,
then $A_{G\otimes H} = A_G \otimes A_H$.

Given any symmetric matrix $M$, we will use 
$\sigma_i(M)$ to denote the $i^{th}$ largest
eigenvalue of $M$. 
For any graph $G$ on $n$-nodes,
we define the spectral radius of $G$, $\rho(G)$,
as the following: 
\[\rho(G)\triangleq
\max_{p\perp \ind{}}
\frac{ \big| p^T A_G p \big| }{\|p\|^2}
= \max( \sigma_2(A_G), |\sigma_{n}(A_G)| ).				
\]
Here $\ind{}$ is the all $1$'s vector of length $n$.
\begin{proposition}\label{thm:large-girth}
  If $H$ is triangle-free, then so does 
  $G\otimes H$.
\end{proposition}
\begin{proof}
  Suppose $G\otimes H$ has a $3$-cycle of the form
  $((a,i), (b,j), (c,k), (a,i))$.
  Then $(i,j,k,i)$ is a closed walk in $H$. 
  $H$ is triangle-free, therefore $i=j$ wlog;
  a contradiction as $H$ has no loops.
\end{proof}
The following Lemma says that as long as $H$ has good spectral
properties, the relative size of maximum independent sets in $G$
will be preserved by $G\otimes H$.
\begin{mylemma}
  \label{thm:kron-is}
  Suppose $H$ is a $d$-regular graph
  with spectral radius $\le \rho$.
  For any graph $G$ with maximum degree
  $\Delta$,
  \[
  \rindset(G\otimes H) \ge
  \rindset(G) \ge
  \Big( 1 - \frac{\rho \Delta}{2 d} \Big)
  \rindset(G\otimes H).
  \]
\end{mylemma}
\begin{proof}
  Suppose $V(G) = [n]$ and $V(H) = [N]$.
  Let $A \triangleq A_G$ be the adjacency matrix of $G$
  and $B$ be the {{normalized}}
  adjacency matrix of $H$,
  \[
  B \triangleq \frac{1}{d} A_H.
  \]
  For the lower bound, consider an independent set
  $I$ in $G$.
  It is easy to check that $I \times [N]$ is an
  independent set in $G\otimes H$, thus
  $\indset(G\otimes H) \ge N \cdot \indset(G)
  \implies
  \rindset(G\otimes H) \ge \rindset(G)$.

  For the upper bound,
  consider the indicator vector $f \in \{0,1\}^{[n] \times [N]}$
  of an independent set in $G\otimes H$.
  Since the corresponding set contains no edges from
  $G\otimes H$,
  \[
  f^T (A \otimes B) f = 0.
  \]
  Define $p \in [0,1]^V$ as the following vector:
  \[
  p_u \triangleq \frac{1}{N} \sum_{j\in N} f_{u,j}.
  \]
  For each $u\in [n]$, pick $u$ with probability $p_u$.
  Let $I_0 \subseteq [n]$ be the set of picked nodes.
  Next, start with $I \gets I_0$.
  As long as there is an edge of $G$ contained in ${I}$,
  arbitrarily
  remove one of its endpoints  from $I$.
  At the end of this process, the remaining
  set $I$ is an independent set for $G$,
  and its size is at least the size of $I_0$ minus
  the number of edges contained in $I_0$.
  Hence
  $
  {}{|I|} \ge {}{|I_0| - |E_G(I_0, I_0)|}.
  $
  Observe that
  \[\expcts{}{|I_0|}
  = \sum_u p_u = \frac{1}{N} \|f\|^2. \quad
  \text{(since $f$ is a $\{0,1\}$ vector, its $l_1$ norm is the same
    as the $l_2$ norm.)}\]
  The probability of any pair $i\neq j$ being contained in $I_0$ is
  given by
  \begin{align*}
    \probs{}{\{i,j\}\subseteq I_0} = &\ \ p_u p_v.
                                       \intertext{Therefore, the expected number of edges
                                       contained in $I_0$ is $\frac12 p^T A p$:}
                                       \expcts{}{|E_G(I_0,I_0)|} = & \sum_{u<v} A_{uv} p_u p_v
                                                                     = \frac12 p^T A p. \\
    =& \frac{1}{2 N} f^T
       (A\otimes \widetilde{J_N}) f
                                                                   & \mbox{(\ref{thm:kron-is:0})}\\
    \le & \frac{1}{2 N} \Big| f^T
          A\otimes (\widetilde{J_N}-B) f
          \Big|
                                                                   & \mbox{(\ref{thm:kron-is:1})}\\
    \le & \frac{\rho \Delta}{2 N d} \|f\|^2.
                                                                   & \mbox{(\ref{thm:kron-is:2})}
  \end{align*}
  Putting it all together:
  \begin{align*}
    \expcts{}{|I|} \ge \expcts{}{|I_0|} -
    \expcts{}{|E_G(I_0, I_0)|}
    = & \frac1N \|f\|^2 - \frac12 p^T A p
        \ge \frac{\|f\|^2}{N} \Big(1 - \frac{\rho \Delta}{2 d}
        \Big).
  \end{align*}
  Therefore
  \[
  \indset(G) \ge \frac{1 - \frac{\rho \Delta}{2 d} }{N}
  \indset(G\otimes H)
  \implies
  \rindset(G) \ge
  \Big(1 - \frac{\rho \Delta}{2 d}\Big) \rindset(G\otimes H).
  \tag*{\qedhere}
  \]
\end{proof}
In the remaining part,
we prove the supporting claims.
\begin{claim}\label{thm:kron-is:0}
  $p^T A p
  = \frac{1}{N}
  {f}^T (A \otimes \widetilde{J_N}) f$ where
  $\widetilde{J_N}$ is the $N$-by-$N$ matrix of all $1/N$'s.
\end{claim}
\begin{proof}
  Let $\ind{}^{u,v} \in \R^{V\times V}$
  be the matrix whose entry at $u^{th}$ row
  and $v^{th}$ column is $1$, and all others $0$.
  Notice $A = \sum_{u,v} A_{u,v} \ind{}^{u,v}$.
  Let $J_N$ be the $N$-by-$N$ matrix of all $1$'s.
  For any pair $(u,v)\in V^2$, 
  \begin{align*}
    p_u p_v = & \frac{1}{N^2} \sum_{i,j} f_{u,i} f_{v,j}
                = \frac{1}{N^2} f^T \Big(\ind{}^{u,v} \otimes J_N
                \Big) f
                =
                \frac{1}{N} f^T \Big(\ind{}^{u,v} \otimes
                \widetilde{J_N}
                \Big) f. \\
    p^T A p = &
                \frac{1}{N}
                \sum_{u,v} A_{u,v}
                f^T \Big(
                \ind{}^{u,v} \otimes \widetilde{J_N} \Big) f
                = \frac{1}{N} f^T \Big[
                \big(\sum_{u,v} A_{u,v} \ind{}^{u,v}\big)
                \otimes \widetilde{J_N} \Big] f \\
    =& \frac{1}{N} f^T (A\otimes \widetilde{J_N}) f.
  \end{align*}
  The second-to-last
  identity follows from the bi-linearity of Kronecker
  product.
\end{proof}
\begin{claim}\label{thm:kron-is:1}
  $f^T (A \otimes \widetilde{J_N}) f
  \le |f^T A \otimes (B - \widetilde{J_N}) f|$.
\end{claim}
\begin{proof}
  We have $f^T (A \otimes \widetilde{J_N}) f
  = f^T \Big[ A \otimes \big(\widetilde{J_N}-B)\Big] f
  + f^T (A\otimes B) f$.
  As noted above, $f$ being an independent
  set implies $f^T (A\otimes B) f = 0$:
  \[
  f^T (A \otimes \widetilde{J_N}) f
  = f^T \Big[ A \otimes \big(\widetilde{J_N}-B)\Big] f		
  \le |f^T A \otimes (\widetilde{J_N}-B) f|.
  \tag*{\qedhere}
  \]
\end{proof}
\begin{claim}\label{thm:kron-is:2}
  $|f^T A \otimes (B - \widetilde{J_N}) f| \le
  \frac{\Delta \rho}{d} \|f\|^2$.
\end{claim}
\begin{proof}
  Define $C \triangleq B - \widetilde{J_N}$.
  For any  
  symmetric matrix $M$, let $\rho(M)$ be its 
  spectral radius, 
  $
  \rho(M) \triangleq \max_{p} \frac{|p^T M p|}{\|p\|^2}.
  $ Observe that $\rho(M) = 
  \max(|\sigma_{i}(M)|)$.
  We have:
  \[
  |f^T A \otimes (B - \widetilde{J_N}) f| 
  = |f^T A \otimes C f| 
  \le \rho(A\otimes B) \|f\|_2^2.
  \] 
  We know that the spectrum of 
  the Kronecker product of two symmetric matrices 
  correspond to the pairwise product of the spectrum of 
  corresponding matrices, i.e., 
  all eigenvalues of $A\otimes C$
  are of the form $\sigma_i(A) \cdot \sigma_j(C)$
  for each $i$ and $j$.
  Therefore, 
  \[\rho(A\otimes C) = \max( |\sigma_i(A) \sigma_j(C)| )
  = \max(|\sigma_i(A)|) \max(|\sigma_j(C)|)
  = \rho(A)\cdot \rho(C).
  \]
  Observe that $\rho(A) \le \Delta$, 
  since $A$ is the adjacency matrix of a graph with 
  degree $\le \Delta$.
  Now we will upper bound $\rho(C)$.
  Since $H$ is a regular graph
  and $B$ is its normalized adjacency matrix,
  the largest eigenvector of $B$ is all $1$'s
  and the corresponding eigenvalue is $1$.
  Therefore
  $C$ has the same eigenspace with $B$.		
  Moreover $C \ind{} = 0$, thus:
  \begin{align*}
    \rho(C) = & \max(|\sigma_i(C)|\ :\ 1\le i\le n)
                = \max(|\sigma_i(B)|\ :\ 2 \le i \le n)\\
    = & \max(\sigma_2(B), |\sigma_n(B)|) 
        = \frac{1}{d} \rho(G).
        \tag*{\qedhere}	
  \end{align*}
\end{proof}
%

We now prove the main theorem needed for our reduction.	
\begin{theorem}
  \label{thm:sparse-vc-to-tri-free}
  Given a graph $G=(V,E)$ with maximum degree $\Delta$,
  for any $\eps > 0$,
  we can
  construct in polynomial time,
  a triangle-free graph
  $\widehat{G}=(\widehat{V}, \widehat{E})$ with
  \[\rindset(G) \le \rindset( \widehat{G} ) \le
  (1+\eps) \rindset(G).\]
  Moreover $\widehat{G}$ has
  \begin{inparaenum}[(a)]
  \item $\mathrm{poly}(\Delta, {\eps}^{-1}) |V|$ nodes,
  \item degree
    $O\big(\displaystyle \Delta^3 \eps^{-2} \big)$.
  \end{inparaenum}
\end{theorem}
\begin{proof}
  For any $d$ and $N$,
  it is known how to 
  construct~\cite{lps88,mor94} in deterministic
  polynomial time,
  a $O(d)$-regular 
  Ramanujan graph $H$ with
  girth $\Omega(\log_d N)$ and
  spectral radius at most $\rho \le O(\sqrt{d})$.
  Thus for some choice of $d = O( \Delta^2 \eps^{-2} )$
  and $N = d^{O(1)} = \mathrm{poly}(\Delta,\eps^{-1})$,
  we can find a $d$-regular graph
  $H$ with girth at least $\Omega(1)$
  and spectral radius
  $\rho \le d \eps / \Delta$.
  For such $H$, let $\widehat{G} \gets G \otimes H$.
  We have
  $\Big(1 - \frac{\rho \Delta}{2 d} \Big)^{-1}
  \le \big(1 - \eps/2 \big)^{-1}
  \le 1 + \eps
  $.
  \ref{thm:large-girth} implies
  $G\otimes H$ is triangle free.
  By~\ref{thm:kron-is}:
  \[\rindset(G) \le \rindset( G \otimes H ) \le
  \Big(1 - \frac{\rho \Delta}{2 d} \Big)^{-1} \rindset(G)
  \le \Big(1 + \eps \Big) \rindset(G).
  \]
  Now we prove the remaining properties:
  \begin{enumerate}[(a)]
  \item $|V(G\otimes H)| = |V(G)| \cdot |V(H)|
    \le |V| \cdot \mathrm{poly}(\Delta,\eps^{-1})$.
  \item $d_{\max}(G\otimes H)
    \le
    d_{\max}(G) \times d_{\max}(H)
    \le O(\Delta d) = O(\Delta^3 \eps^{-2})$.
    \qedhere
  \end{enumerate}
\end{proof}
\begin{note}
  Noga Alon has provided an alternate construction where one can
  obtain a triangle free graph $\hat{G}$ such that
  $\rindset(\hat{G}) = \rindset(G)$. This however, does not lead to
  improved constant in our analysis. For the sake of completeness, we
  include the alternate theorem in the Appendix~(See
  Theorem~\ref{thm:noga-construction}).
\end{note}
We will end the section with the proof
of~\ref{thm:vc-tri-free-hardness}.  We need the following hardness
result of~\cite{ds05}: It follows from their Corollary 2.3 and
Appendix 8 (weighted to unweighted reduction).  As noted
in~\cite{ds05}, the construction produces bounded degree graphs.
\begin{theorem}[Dinur, Safra~\cite{ds05}]
  \label{thm:vc-hard-ds}
  For any constant $\eps > 0$,
  given any unweighted graph $G$ with bounded degrees,
  it is NP-hard to distinguish between:
  \begin{itemize}
  \item (Yes) $\rindset(G) > c - \eps$,
  \item (No) $\rindset(G) < s + \eps$; 
  \end{itemize} where $c$ and $s$ are constants 
  such that $\frac{1-s}{1-c} \approx 1.36$.
\end{theorem}
\begin{proof}[Proof of~\ref{thm:vc-tri-free-hardness}]
  Given a bounded degree
  graph $G$, consider the graph $\widehat{G}$ given
  by~\ref{thm:sparse-vc-to-tri-free} for some small
  constant $\eps_0 < \eps$.
  Since $G$ is bounded degree and $\eps_0$ is constant,
  $\widehat{G}$ is also bounded degree.
  Furthermore, $\widehat{G}$ satisfies
  $\rindset(G) \le \rindset(\widehat{G}) \le
  (1+\eps_0) \rindset(G)$.
  Completeness follows immediately: $\rindset(\widehat{G})
  > c - \eps$.
  For the soundness, suppose
  $\rindset(\widehat{G}) > s + \eps$.
  Then $\rindset(G) \ge \frac{s+\eps}{1+\eps_0}
  \ge s + \eps$ for suitable $\eps_0$.
  The hardness of \minvc\ follows from \ref{thm:vc-is-dual}.
\end{proof}

\section{Conclusions}
In this paper we provide the first hardness of approximation for the
fundamental Euclidean $k$-means problem. Although our work clears a
major hurdle of going beyond NP-hardness for this problem, there is
still a big gap in our understanding with the best upper bound being a
factor $(9+\epsilon)$. We believe that our result and techniques will
pave way for further work in closing this gap. Our reduction from
vertex cover produces high dimensional instances~($d = \Omega(n)$) of
$k$-means. However, by using the Johnson-Lindenstrauss
transform~\cite{dasgupta2003elementary}, we can project the instance
onto $O(\log n/\epsilon^2)$ dimensions and still preserve pairwise
distances by a factor $(1+\epsilon)$ and the $k$-means cost by a
factor of $(1+\epsilon)^2$. We leave it as an open question to
investigate inapproximability results for $k$-means in constant
dimensions. It would also be interesting to study whether our
techniques give hardness of approximation results for the Euclidean
$k$-median problem. Finally, our hardness reduction
in~\ref{sec:vc-girth} provides a novel analysis by using the spectral
properties of the underlying graph to argue about independent sets in
graph products -- this connection could have applications beyond the
present paper.

\section{Acknowledgments}
We would like to thank Noga Alon and Oded Regev for valuable feedback
on the results in \ref{sec:vc-girth}, in particular for suggesting
alternate proofs of Theorem~\ref{thm:large-girth} and
Theorem~\ref{thm:kron-is}. We would also like to thank Noga for
pointing out that the graph product construction in \ref{sec:vc-girth}
does not eliminate even cycles.

\bibliographystyle{plain}
\bibliography{paper}
\appendix

\section{Appendix}
\begin{theorem}
  \label{thm:noga-construction}
  Let $G = (V,E)$ be an arbitrary graph with maximum degree
  $\Delta$. It is possible to construct in polynomial time a triangle
  free graph $\hat{G}$ such that $\rindset(\hat{G}) = \rindset(G)$.
\end{theorem}
Before proving the theorem, we need the following standard facts about
$(n,d,\lambda)$ graphs. The following proofs are suggested by Noga
Alon.
\begin{mylemma}
  \label{lem:facts-about-expanders}
  Let $H=(U,F)$ be an $(n,d,\lambda)$ graph, assume $\lambda < d/4$
  and let $B$ be a set of vertices of $H$. Let $N(B)$ denote the set
  of all neighbors of $B$ in $H$. Then:
  \begin{enumerate}
  \item If $|B| > \frac{\lambda}{d}n$ then $|N(B)| > n-\frac{\lambda}{d}n$
  \item If $|B| \leq \frac{\lambda}{d}n$ then $|N(B)| \geq \frac{\lambda}{2d}n$
  \end{enumerate} 
\end{mylemma}
\begin{proof}
  Part $1$ is proved in Corollary $1$ in
  \cite{alon1992construction}. Part $2$, for
  $\frac{2\lambda^2}{d^2}n \leq |B| \leq \frac{\lambda}{d}n$ follows
  from the same corollary (which implies that in this range
  $|N(B)| \geq \frac n 2$). For $|B| \leq \frac{2\lambda^2}{d^2}n$,
  the result follows from the expander mixing lemma (see
  \cite{alon2004probabilistic}, corollary $9.2.5$), as there are
  $d|B|$ edges between $B$ and $N(B)$.
\end{proof}
We now provide the proof of Theorem~\ref{thm:noga-construction}.
\begin{proof}
  Let $H = (U,F)$ be a $(n,d,\lambda)$-expander with
  $\lambda \leq
  2\sqrt{d-1}$\footnote{This
    means that all eigenvalues of $H$, except the first, are bounded
    by $lambda$.}
  Let $\hat{G} = G \otimes H$. Further, let
  $\frac d {2\lambda} \geq \Delta$. It is well known that such graphs
  exist. It is easy to see that any
  $\rindset(G \otimes H) \geq \rindset(G)$, since any independent set
  $S$ in $G$ leads to an independent set $S \otimes U$ in
  $G \otimes H$.

  For the other direction, let $S \subset V \times U$ be an
  independent set in $G \otimes H$. Define
  \[
  T = \{ v \in V: |\{u \in U: (v,u) \in S\}| \geq \frac{\lambda}{d}n\}
  \]
  Be Lemma~\ref{lem:facts-about-expanders} part $1$, $T$ is an
  independent set in $G$. Let $T'$ be a maximal (with respect to
  containment) independent set in $G$ that contains $T$. By
  maximality, every vertex in $V\setminus T'$ has at least one
  neighbor in $T'$. Thus $T'$ is a dominating set in $G$ and there is
  a collection of stars $\{S_v: v \in T'\}$, covering all the vertices
  of $G$. As $T'$ is an independent set, $|T'| \leq
  \rindset(G)|V|$.
  To complete the proof it suffices to show that for each of the stars
  $S_v$ in our collection whose set of vertices in $G$ is $V_v$
  \begin{eqnarray}
    \label{app:eq1}
    |\{(v',u): (v',u) \in S, v' \in V_v\}| \leq |U| = n
  \end{eqnarray}
  The number of leaves of the star $S_v$ is at most $\Delta$. For each
  such leaf $v'$, the set of vertices of $H$ given by
  \[
  B_{v'} = \{u \in U: (v',u) \in S\}
  \]
  is of cardinality smaller than $\frac{\lambda}{d}n$. Moreover, all
  its neighbors in $H$ cannot belong to the set
  $B_v = \{u \in U: (v,u) \in S\}$ where $v$ is the center of the star
  $S_v$. By Lemma~\ref{lem:facts-about-expanders} part $2$, the number
  of these neighbors is at least $\frac{d}{2\lambda} \geq \Delta$
  times the cardinality of $B_{v'}$. This implies that the total size
  of all sets $B_{v'}$ where the sum ranges over all leaves $v'$ of
  $S_v$ is at most the number of vertices in $U - B_v$, implying
  \ref{app:eq1} and completing the proof.
\end{proof}

\end{document}